\documentclass{article}

\usepackage[preprint]{neurips_2023}

\usepackage[utf8]{inputenc} %
\usepackage[T1]{fontenc}    %
\usepackage{hyperref}       %
\usepackage{url}            %
\usepackage{booktabs}       %
\usepackage{amsfonts}       %
\usepackage{nicefrac}       %
\usepackage{microtype}      %
\usepackage{xcolor}         %
\usepackage{amsmath}
\usepackage{amssymb}
\usepackage{subcaption}
\usepackage{graphicx}
\usepackage{amsthm}

\newtheorem{prop}{Proposition}
\newtheorem{rmk}{Remark}

\newcommand{\KL}{\operatorname{KL}}

\newcommand{\Prob}{\mathrm{Prob}}
\newcommand{\Kc}{\mathcal{K}}
\newcommand{\Gc}{\mathcal{G}}
\newcommand{\Rc}{\mathcal{R}}
\newcommand{\Xc}{\mathcal{X}}
\newcommand{\Cc}{\mathcal{C}}

\newcommand{\Nc}{\mathcal{N}}
\newcommand{\diff}{\mathrm{d}}
\newcommand{\ones}{\mathbf{1}}

\newcommand{\rev}[1]{#1}
\renewcommand{\vec}[1]{\mathbf{#1}}

\title{Joint trajectory and network inference via reference fitting}

\author{%
  Stephen Y. Zhang  \\ 
  School of Mathematics and Statistics, University of Melbourne \\ 
  Department of Genetics, Stanford University \\
  \texttt{stephenz@student.unimelb.edu.au}
}

\begin{document}

\maketitle

\begin{abstract}
  Network inference, the task of reconstructing interactions in a complex system from experimental observables, is a central yet extremely challenging problem in systems biology.
  While much progress has been made in the last two decades, network inference remains an open problem.
  For systems observed at steady state, limited insights are available since temporal information is unavailable and thus causal information is lost.
  Two common avenues for gaining \emph{causal} insights into system behaviour are to leverage temporal dynamics in the form of trajectories, and to apply interventions such as knock-out perturbations.
  We propose an approach for leveraging \emph{both} dynamical and perturbational single cell data to jointly learn cellular trajectories and power network inference.
  Our approach is motivated by min-entropy estimation for stochastic dynamics and can infer directed and signed networks from time-stamped single cell snapshots. 
\end{abstract}

\section{Introduction}

Cells are complex systems which are able to process and respond to molecular signals. A coarse but helpful simplification that lies at the heart of much of systems biology is to think of cells as a collection of interacting molecular species, and cellular behaviour as emerging from the dynamics of this molecular circuit. Viewing cells as dynamical systems poses the inverse problem of recovering information about the structure of the underlying interaction network from experimental observables. While network inference has received much attention across the span of the last two decades \cite{stumpf2021inferring, marbach2012wisdom} it remains largely an open problem, and real biological networks remain poorly characterised with only a few exceptions. As technological advances continue to push the limits of what can be measured in experiment, opportunities are created for inference methods to leverage new modalities of data \cite{badia2023gene}.

Two widely adopted experimental paradigms in modern single cell biology are time-resolved single cell transcriptomics and single cell perturbation assays. Many important biological processes, notably development, are characterised by a temporal evolution of a population of cells. Time-series studies allow population-level observation of this evolution via serial independent sampling across several timepoints. Importantly, the destructive nature of measurement means longitudinal tracking is not possible, so individual trajectories must be reconstructed \cite{haghverdi2016diffusion, schiebinger2019optimal}. Observation of temporal behaviour of the system in its natural state makes it possible to distinguish between cause and effect, and a range of computational methods have been developed to infer \emph{directed} networks from time-series single cell data \cite{ding2020analysis}.

On the other hand, perturbational studies allow \emph{interventions} such as gene knockouts to be applied to the biological system of interest in order to study the system's behaviour outside of its natural state. Interventions are a powerful approach for studying causal mechanisms underlying observed data \cite{pearl2010introduction}. Recent technologies have made large scale gene knockout/knockdown studies possible \cite{dixit2016perturb, yao2023scalable}, and the task of utilising these data for powering network inference arises naturally \cite{fiers2018mapping}. This direction has received increasing attention recently \cite{dixit2016perturb, yang2020scmageck, ishikawa2023renge, rohbeck2024bicycle}, but existing analysis approaches (with exception of \cite{ishikawa2023renge}) have predominantly focused on settings where only steady-state measurements are available from the perturbed and non-perturbed systems. In developmental systems, transient dynamics play a crucial role in determining cell fate and are thus rich in information about the governing principles that drive observed dynamics \cite{teschendorff2021statistical, reid2018transdifferentiation, maclean2018exploring}. 

We propose an approach to jointly infer trajectories and interaction networks from time-series single cell data, which we call \emph{reference fitting}.
Drawing inspiration from trajectory inference approaches based on entropy-regularised optimal transport and the Schr\"odinger bridge \cite{lavenant2024toward, schiebinger2019optimal}, our method is motivated by a \emph{least-action principle}: the observed trajectory taken by a dynamical system should minimise an energy relative to a \emph{reference process}, which depends on the system structure \cite{heitz2021ground}. While in the trajectory inference setting this reference was taken to be uninformative (i.e. no prior structural information), we now consider a parametric family of \emph{linear} reference processes and search for one which is action minimising, given observations \cite{freidlin1998random}.
Our approach can be applied to time-series datasets capturing the natural evolution of a system of interest, as well as perturbation data in the form of gene knockouts to improve the inference results. In our view, this is a major advantage of our approach over many existing trajectory and network inference methods which cannot leverage perturbation information.
Using simulated systems ensures a objective and unbiased assessment of inference performance, and we find that availability of perturbation data greatly improves the inference accuracy even if only a subset of genes are perturbed. We demonstrate the application of our method to the time-series human induced pluripotent stem cell (hiPSC) CRISPR knockout dataset of \cite{ishikawa2023renge} and find that the inferred networks compare favourably to a ChIP-seq reference subnetwork and agree with biological prior knowledge. 

After the initial version of this work was completed, the author became aware of a concurrent study addressing the problem of iteratively fitting a reference process for dynamical inference \cite{shen2024multi}. These two works differ in aspects of their application setting and implementation details -- our work uses the static formulation of Schr\"odinger bridges and an Ornstein-Uhlenbeck reference family to study interventions, while Shen et al. \cite{shen2024multi} employ a dynamical formulation and neural family of drift. However, the underlying idea is the same --  to depart from a fixed reference process and iteratively fit both the couplings and reference. 

\begin{figure}
  \centering 
  \includegraphics[width = \linewidth]{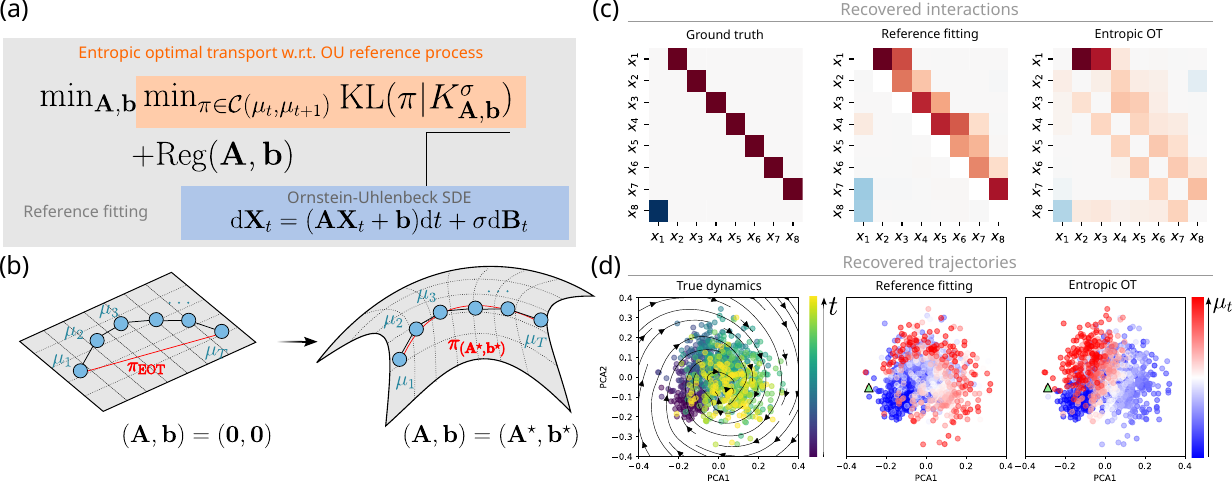}
  \caption{\rev{(a) Entropic optimal transport and reference fitting with an Ornstein-Uhlenbeck (OU) family. (b) Given a series of observed population snapshots and starting with a pure Brownian reference, iterative fitting of the reference process allows to progressively improve an estimate of the underlying dynamics. (c) Ground truth and recovered interactions for a 8-dimensional non-equilibrium OU process. (d) Temporal dynamics inferred by reference fitting and standard entropic OT, as well as the true vector field shown in the two leading principal components. In the right two panes, the family of marginals starting from a fixed point (green triangle) are shown.}}
  \label{fig:reference_fitting_OU}
\end{figure}

\section{Methods}

\paragraph{Dynamical inference}
\rev{We model cell state $X_t \in \mathbb{R}^d$ with an autonomous, drift-diffusion stochastic dynamics driven by Brownian noise $B_t$ with intensity $\sigma^2$:}
\begin{equation}
  \diff X_t = f(X_t) \: \diff t + \sigma \: \diff B_t, \quad X_0 \sim \rho_{0} .
  \label{eq:generic_sde}
\end{equation}
Consider a time-series observation setting, where snapshots from the process at $T \ge 2$ consecutive timepoints $0 = t_1, \ldots, t_T = 1$ are drawn. That is, the snapshot at each time $t_i$ comprises a collection of $N_i$ independently measured cell states $\Xc_i = \{ x_j^{t_i} \}_{j = 1}^{N_i} \subset \mathbb{R}^d$. Equivalently, this can be represented as an empirical distribution $\hat{\mu}_i = \frac{1}{N_i} \sum_{j = 1}^{N_i} \delta_{x_j^{t_i}}$. \rev{Importantly, we consider the case where longitudinal measurement is infeasible, as is the case for high-throughput single cell experiments \cite{deconinck2021recent, schiebinger2019optimal}.}

Imposing a Brownian reference dynamics where the noise level $\sigma > 0$ is known and fixed, between two consecutive snapshots $(\mu, \mu')$ at instants $t = 0, 1$, a well known entropic least action principle corresponding to the Schr\"odinger bridge \cite{leonard2013survey} can be used to infer the ``most likely'' conditional evolution of the system:
\begin{equation}
  \min_{\pi \in \Cc(\mu, \mu')} \sigma^2 \KL(\pi | K_\sigma),
  \label{eq:least_action_brownian}
\end{equation}
where $\KL(\mu | \nu) = \int \diff \mu \log \left( \diff \mu / \diff \nu \right)$ denotes the Kullback-Leibler divergence between probability measures. \rev{In the above, $K_\sigma$ is the Gaussian transition kernel on $\mathcal{X} \times \mathcal{X}'$, i.e. $K_{\sigma}(x_i, x_j) \propto \exp\left( - \| x_i - x_j \|_2^2 / 2\sigma^2 \right)$, and the coupling $\pi \in \Prob(\Xc \times \Xc')$ describes inferred trajectories of cells between successive timepoints.}
\rev{The set of candidate couplings is
  \begin{equation}
    \Cc(\mu, \mu') = \Biggl\{ \gamma \in \Prob(\Xc \times \Xc') : \sum_j \gamma_{ij} = \mu_i, \sum_i \gamma_{ij} = \mu'_j \Biggr\}.
  \end{equation}
  That is, the set of all possible joint distributions compatible with the marginals $(\mu, \mu')$.}
The least action principle \eqref{eq:least_action_brownian} can therefore be understood also as a minimum-entropy principle, where the most likely conditional evolution $\pi$ is the one that is closest to the reference process $K_{\sigma}$ in relative entropy \cite{leonard2013survey}. 

\paragraph{Reference fitting}
Instead of using a Brownian motion reference process $\sigma B_t$ (which specifies a prior dynamics with only diffusion), we consider a more general family of Ornstein-Uhlenbeck (OU) processes described by the linear SDE 
\begin{equation}
  \diff X_t = (A X_t + b) \: \diff t + \sigma \: \diff B_t.
  \label{eq:linear_sde}
\end{equation}
\rev{These dynamics exhibit both drift and diffusion, where the drift component is prescribed by a linear interaction matrix $A \in \mathbb{R}^{d \times d}$, as well as a constant drift term $b$.} We do not impose any structural constraints on $A$, allowing this model to capture directed and signed interactions. Each element $A_{ij}$ can thus be interpreted as the effect of gene $j$ on gene $i$, where a positive (negative) value means activation (repression). Systems of the form \eqref{eq:linear_sde} naturally arise as the linearisation of more complex stochastic systems -- expanding \eqref{eq:generic_sde} about $X = 0$ for instance yields \eqref{eq:linear_sde} with $A_{ij} = \frac{\partial f_i}{\partial x_j}$. Thus, while realistic biological dynamics exhibit complex behaviour consistent with non-linear systems, the linearised dynamics we consider provide a middle ground between biophysical realism and mathematical tractability. 

Let $K_{A, b}^{\sigma}$ be the transition kernel of the OU process \eqref{eq:linear_sde} with drift given by $(A, b)$. We consider the problem where both the coupling $\pi$ and reference $K_{(A, b)}^{\sigma}$ are sought:
\begin{equation}
  \min_{A \in \mathbb{R}^{d \times d}, b \in \mathbb{R}^{d}} \: \min_{\pi \in \Cc(\mu, \mu')} \sigma^2 \KL(\pi | K_{(A, b)}^{\sigma}) + \Rc(A, b).
  \label{eq:least_action_ou}
\end{equation}
\rev{In the above, $\Rc$ is a regulariser applied to $(A, b)$. As we explain in what follows, this is essential to ensure a well-defined optimisation problem due to issues of non-identifiability of the drift matrix $A$ from snapshot observations.}

\rev{Clearly, if we discard the requirement that the reference kernel $K$ arise from a SDE of the form \eqref{eq:linear_sde}, there are, unhelpfully, \emph{infinitely} many pairs $(\pi, K)$ that satisfy $\KL(\pi | K) = 0$. Given some $\pi \in \Cc(\mu, \mu')$ one may trivially pick $K = \pi$.
  On the other hand, constraining $K_{(A, b)}^{\sigma}$ to be the transition kernel for \eqref{eq:linear_sde} for some parameters $(A, b)$ provides necessary additional structure to avoid these trivial solutions. Writing $\Kc$ to denote the set of feasible reference kernels (in our case, the family of OU transition kernels generated by some $(A, b)$), if $\Kc \cap \Cc(\mu, \mu')$ is non-empty then there exists at least one process of the form \eqref{eq:linear_sde} that perfectly explains the observations $(\mu, \mu')$. On the other hand if this intersection is null, we are essentially seeking the ``closest'' (in terms of KL-divergence) pair of coupling and transition kernel. }

\paragraph{Modelling perturbations} Although our framework is applicable to time-series snapshot data in general, we highlight the setting where time series data with perturbations are available, as they may help resolve interactions that are not identifiable from the natural dynamics alone. Motivated by recent works \cite{ishikawa2023renge, rohbeck2024bicycle}, we consider gene knock-out perturbations, in which expression of a gene of interest may be switched off. \rev{In what follows we omit the bias $b$ and focus on fitting the interaction matrix $A$, although we remark that generalisation to affine drifts is straightforward.}

In a scenario where a gene $g$ is knocked out, we \emph{modify} the linear interaction matrix to $A^{(g)}$ where the $g$th row is set to zero, reflecting that the expression of the knocked-out gene $g$ is no longer dependent upon other genes. That is, $A^{(g)} = A \odot M^{(g)}$ where $M^{(g)}_{ij} = \ones_{i \ne g}$ is a masking matrix. For a gene $g$ knockout, one then has the following modified reference dynamics:
\[
  \diff X^{(g)}_t = (A \odot M^{(g)}) X^{(g)}_t \: \diff t + \sigma \: \diff B_t . 
\]

For a set of perturbed genes $\Gc$ along with the wild type trajectory observed at times $t = 1, \ldots, T$, we formulate the following objective over interaction matrices $A$ of the OU process
\rev{\begin{align}
  \min_{A} \frac{1}{|\Gc|+1} \sum_{g \in \Gc \cup \{ \emptyset \}} \left[ \frac{1}{T-1} \sum_{i = 1}^{T-1} \min_{\pi^{(g)}_i \in \Cc(\hat{\mu}^{(g)}_i, \hat{\mu}^{(g)}_{i+1})} \sigma^2 (t_{i+1} - t_i) \KL(\pi^{(g)}_{i} | K^{(g)}_{A, \sigma}(t_{i+1} - t_i)) \right] + \lambda \Rc(A).
  \label{eq:objective_joint}
\end{align}}
\rev{In the above, $\pi^{(g)}_i$ are the couplings between times $t_i, t_{i+1}$ for condition $g$, and $K^{(g)}_{A, \sigma}(\Delta t)$ corresponds to the reference kernel under perturbation condition $g$ over a time interval $\Delta t$}. We have written $g = \emptyset$ to correspond to the wild type. The functional $\Rc$ applies a regularisation to the interaction matrix $A$ which is crucial for recovering good results in this non-convex inference problem \rev{(see the appendix for a discussion)}. In practice, we use an elastic net regulariser \cite{zou2005regularization}
\begin{equation}
  \Rc(A) = \alpha \| A \|_2^2 + (1-\alpha) \| A \|_1, \: \alpha \in [0, 1].
  \label{eq:elastic_net}
\end{equation}

We remark that, without perturbation or temporal information, one cannot expect to recover the interaction matrix $A$ from snapshot data since there may be many matrices $A$ which give rise to the same equilibrium distribution \cite{rohbeck2024bicycle, dettling2023identifiability}. In our setting, temporal information in the form of snapshots of transient states, as well as perturbations may help resolve these non-identifiabilities.

\paragraph{Identifiability and necessity of regularisation}

We note that, for a \emph{fixed} reference process $K_{A}^{\sigma}$ , the problem \eqref{eq:least_action_ou} is convex in $\pi$ with a unique minimiser. Furthermore, the optimal $\pi$ can be computed via the celebrated Sinkhorn matrix scaling algorithm \cite{cuturi2013sinkhorn}.
\rev{The problem in terms of $K_{A}^\sigma$, however, is non-convex since the transition kernel depends upon $A$ via a matrix exponential. This is in some sense a consequence of fundamental limitations of drift inference from snapshot data \cite{weinreb2018fundamental}: in general, non-conservative forces cannot be uniquely recovered from snapshot observations alone. As a simple two-dimensional counterexample, one may consider $X_0$ distributed with radial symmetry and $A = \kappa \left[\begin{smallmatrix} 0 && 1 \\ -1 && 0 \end{smallmatrix}\right] - \mathrm{Id}$ for any $\kappa \in \mathbb{R}$. 
  In the objective \eqref{eq:least_action_ou}, this manifests through the non-injectivity of the matrix exponential $A \mapsto e^{tA}$ upon which $K_{(A, b)}^\sigma$ depends, whenever $A$ is asymmetric with complex eigenvalues (which is exactly the case we are interested in). In the context of OU processes, this issue has been recognised and discussed in the steady-state setting of graphical continuous Lyapunov models \cite{fitch2019learning, varando2020graphical} for which the question of theoretical consistency remains, to the best of this author's knowledge, open \cite{dettling2024lasso}. In all of these instances, the use of Lasso regularisation has been key to deal with the non-identifiability issues and achieve good results.
  In our setting also, we find that penalisation of $(A, b)$ is required to ensure well-posedness of \eqref{eq:least_action_ou}, namely existence of local minimisers.}

\paragraph{Approximation of the transition kernel}
For OU processes, transition kernels are Gaussians parameterised by their mean and covariance: $X_t | X_0 = x_0 \sim N(\mu_t, \Sigma_t)$, where $\mu_t = e^{t A} x_0$ and $\Sigma_t = \sigma^2 \int_0^t e^{(t - \tau) A} e^{(t - \tau) A^\top} d\tau$. 
In practice, although closed-form expressions are available for the transition density, we found that numerical optimisation of $A$ while accounting for the covariance structure is unstable. We make an approximation in which we decouple the drift and noise:
\[
  \mu_t = e^{t A} x_0, \quad \Sigma_t = \sigma^2 t I. 
\]
This approximation can be understood as a splitting approximation: the Fokker-Planck equation governing the density evolution of the OU process is
\begin{align*}
  \partial_t u_t(x) = -\nabla \cdot (u \vec{v}(x)) + \frac{\sigma^2}{2} \Delta u. 
\end{align*}
A standard splitting scheme inspired by the numerical PDE literature \cite{holden2010splitting} amounts to applying solution operators for the advection and diffusion steps separately. The solution for the advection step amounts to application of the propagator $e^{tA}$, while the diffusion step corresponds to convolution with a heat kernel of bandwidth $\sigma^2 t$: $K(x, x') \propto \exp\left( - \frac{\| x - x' \|_2^2}{2 \sigma^2 t} \right)$. The overall result of this approximation is a transition kernel of the form
\[
  K_t(x, x') \propto \exp\left( - \frac{ \| e^{t A} x - x' \|_2^2 }{2 \sigma^2 t} \right). 
\]

\paragraph{Interpretation as a ground cost learning problem}
The problem \eqref{eq:objective_joint} is formulated in terms of the unknown couplings $\pi$ between snapshots as well as the interaction matrix $A$ parameterising the Ornstein-Uhlenbeck reference process and is therefore a \rev{joint} optimisation problem. Evaluation of the objective for a given $A$ requires solution of a matrix scaling problem via Sinkhorn iterations at each step. From this point of view, the problem is analogous to the ground metric learning problem, in which a series of probability distributions are given and an underlying metric is sought for which the observed sequence of distributions is action minimising \cite{heitz2021ground}. In a sense, our approach can be viewed as drawing inspiration from metric learning to address system identification, where the metric is tied to the underlying system structure.

In our setting, the cost we consider does not arise from a ground metric \emph{per se}, but rather in a probabilistic sense from the transition kernel of a reference process.
Given an observed sequence of distributions $\hat{\mu}_1, \ldots, \hat{\mu}_T$, we then seek a dynamics $A$ that minimises the overall action $A \mapsto \frac{1}{T} \sum_{t = 1}^{T-1} T_A(\hat{\mu}_t, \hat{\mu}_{t+1})$, where $T_A(\mu_t, \mu_{t+1}) = \min_{\pi \in \Cc(\mu_t, \mu_{t+1})} \KL(\pi | K_t(A))$ is the entropic optimal (EOT) transport cost between $\mu_t$ and $\mu_{t+1}$ with reference $A$.

\paragraph{Interpretation as an inverse optimal transport problem}
Finally, we remark that the problem \eqref{eq:least_action_ou} can be related to a form of inverse optimal transport \cite{stuart2020inverse} as follows. Let $\pi$ be fixed, i.e. calculated with $A = 0$ in which case they coincide with the standard entropy regularised optimal transport couplings of \eqref{eq:least_action_brownian}. It is well known that in this case for a time-series observation setting, the recovered couplings $\pi$ are consistent with the evolution of a potential-driven system \cite{lavenant2024toward}:
\[
  \diff X_t = -\nabla \Psi(X_t) \diff t + \sigma \diff B_t. 
\]
Then one seeks to minimise in $A$ the objective of \eqref{eq:least_action_ou}, with $\pi$ fixed. This can then be seen as that of finding $A$ which induces a cost function for which the action of a given coupling $\pi$ is minimised.

\section{Results}

\paragraph{Reference fitting using transient dynamics}
We first demonstrate the utility of learning a reference process simultaneously with the couplings between captured snapshots, Drawing some inspiration from the example of \cite{chardes2023stochastic}, we use the simple example of a non-equilibrium Ornstein-Uhlenbeck process in 8 dimensions. \rev{We choose the matrix $A$ to have the pattern shown in Figure \ref{fig:reference_fitting_OU}(c) and to be Hurwitz stable, i.e. with all eigenvalues having negative real part.} We set $\sigma = 0.05$ and sampled $T = 10$ independent timepoints of 100 points in the time interval $t \in [0, 10]$, where the initial condition was chosen to be out-of-equilibrium: $e_1 + \Nc(0, 0.05)$. As a result, the system is observed to relax towards its (non-equilibrium) steady state via a transient dynamics, which we capture in our time series (shown in Figure \ref{fig:reference_fitting_OU}(a)).

We apply reference fitting to the sampled snapshots, opting for an alternating optimisation scheme over the couplings and the reference dynamics \rev{(see Appendix for further discussion)}. The inferred interactions are shown in Figure \ref{fig:reference_fitting_OU}(c). We considered also the case where couplings were obtained from \rev{EOT (i.e. being the optimal ones for a Brownian motion prior)}, which in theory are consistent with gradient driven dynamics \rev{and are thus inappropriate to describe this system which exhibits non-conservative forces}. Clearly, the reference fitting approach recovers the underlying pattern in the interaction matrix, while in the case of fixed couplings this is lost.

Our reference fitting yields not only an inferred interaction matrix but also couplings $(\pi_{1}, \ldots, \pi_{T-1})$ adapted to the fitted reference process. These couplings correspond to the inferred dynamics. In Figure \ref{fig:reference_fitting_OU}(d) we illustrate the recovered processes by showing the family of marginals starting from a test point at $t = 0$. That is, we chose a point $x^*$ (shown as the green triangle) at time $t_1 = 0$, and considered $\mathbb{P}(X_{t_i} = \cdot | X_{t_1} = x^\star), 1 \leq i \leq T$. From the marginals traced out by this construction, reference fitting produces a inferred dynamics that agrees with the true drift (clockwise in the first two PCA dimensions) while with fixed couplings the underlying rotational vector field is clearly not captured.

\paragraph{Simulated data with knockouts}

We next consider a more realistic simulation model of gene expression dynamics which captures the nonlinear nature of natural biological networks. We use BoolODE \cite{pratapa2020benchmarking}, a trajectory simulation tool which models cellular dynamics as arising from a boolean network using a chemical Langevin equation (CLE). Cells were simulated from a regulatory network of 8 transcription factors displaying trifurcating structure shown in Figure \ref{fig:trifurcating}(a) at 5 timepoints, with a total of 1000 cells in each time-course. In this network, three branches arise from the activation of $\{ g4, g5, g6 \}$ by $g3$ and mutual repression. The importance of these genes in the network is also reflected in their high network centrality scores. In addition to the unperturbed system we simulate three additional cases \rev{where one of $g3, g4, g6$ have been knocked out (see Figure \ref{fig:trifurcating}(b))}.

\begin{figure}[h]
  \centering
  \includegraphics[width = \linewidth]{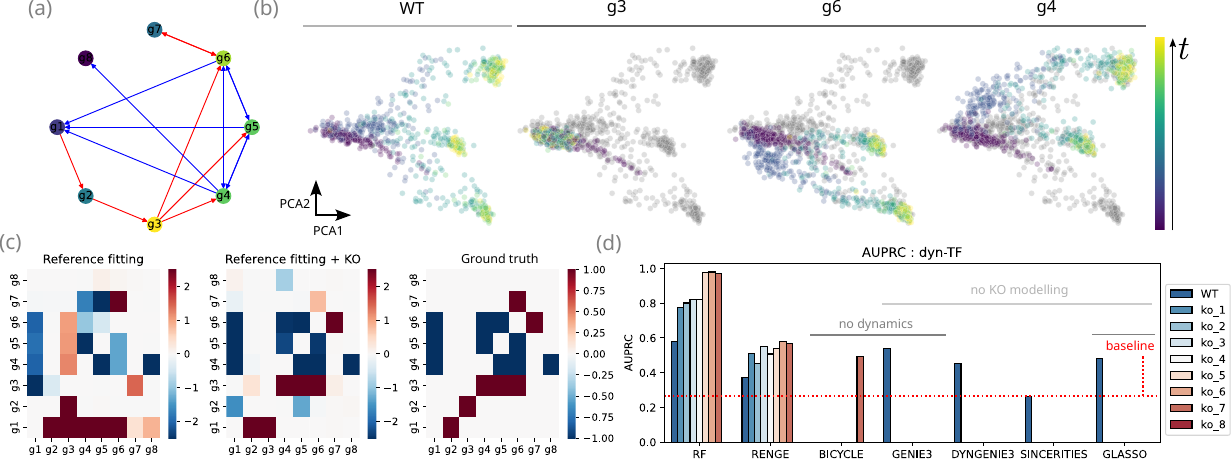}
  \caption{(a) Trifurcating synthetic network, nodes coloured by centrality. Activating (inhibitory) edges are shown in red (blue). (b) Wild-type samples and knockouts, coloured by simulation time. (c) Inferred networks from reference fitting without (AUPR = 0.54) and with knockouts (AUPR = 0.84), compared to the ground truth. (d) Network inference accuracy (averages over 10 datasets) as measured by AUPRC using reference fitting with and without knockouts, compared to alternative methods. }
  \label{fig:trifurcating}
\end{figure}

Applying reference fitting to the snapshot data with knockouts, we find that we achieve good recovery of the network topology (AUPR = 0.84, Figure \ref{fig:trifurcating}(c)). On the other hand when only the unperturbed data are used, performance worsens considerably (AUPR = 0.54).

\rev{To further understand the performance of RF for different numbers of knockouts, we vary the number of available knockouts starting from the WT only trajectory and adding knockout genes in order of decreasing out-edge eigencentrality. We also apply several other network inference approaches for comparison: RENGE \cite{ishikawa2023renge}, which is the only other network inference approach that can utilise both temporal and perturbational data, as well as BICYCLE \cite{rohbeck2024bicycle} which models perturbations but not time-series data. We include the more classical methods GENIE3 \cite{huynh2010inferring} which was shown to be among the top performers in benchmarks \cite{pratapa2020benchmarking}, as well as dynGENIE3 \cite{huynh2018dyngenie3} and SINCERITIES \cite{papili2018sincerities} which are designed for time-series data. Finally as an additional baseline, we consider the graphical LASSO \cite{friedman2008sparse} which is a classical covariance-based network inference approach. Since none of these approaches model perturbations, we ran them on the combined data across all conditions.}

\rev{The results of this comparison are shown in Figure \ref{fig:trifurcating}(d) where performance is again measured in terms of the area under the precision-recall curve (AUPRC), as has been the standard in prior work \cite{pratapa2020benchmarking}. We find that reference fitting with time-series, even in the case \emph{without} perturbations, outperforms competing methods in this case. Additionally, we find that RF performance increases with the number of available knockout conditions. Although RENGE is also designed to leverage temporal and perturbation information, we find that even with knockouts it only marginally outperforms GENIE3 which does not use either temporal or perturbation information. This suggests that RENGE may not be effectively integrating these additional sources of information.
  We find that the accuracy of BICYCLE is also relatively low despite having many knockout conditions available. While this may largely reflect the limitations of steady-state assumptions, we also note that the optimisation scheme for BICYCLE is fairly sophisticated \cite{rohbeck2024bicycle} and that it is possible that hyperparameter tuning or longer training may improve results\footnote{Due to compute time constraints, we ran BICYCLE for each instance on CPU for 5000 epochs pretraining latents and 5000 epochs fitting the model, which took over 12 hours on CPU for each instance. By comparison, RF runs in a matter of minutes for the same input size. }. 
We find that SINCERITIES performs particularly poorly. Similar behaviour was observed in \cite{pratapa2020benchmarking}, where SINCERITIES struggled the most in the trifurcating case. }

\begin{figure}[h]
  \centering
  \includegraphics[width = 0.8\linewidth]{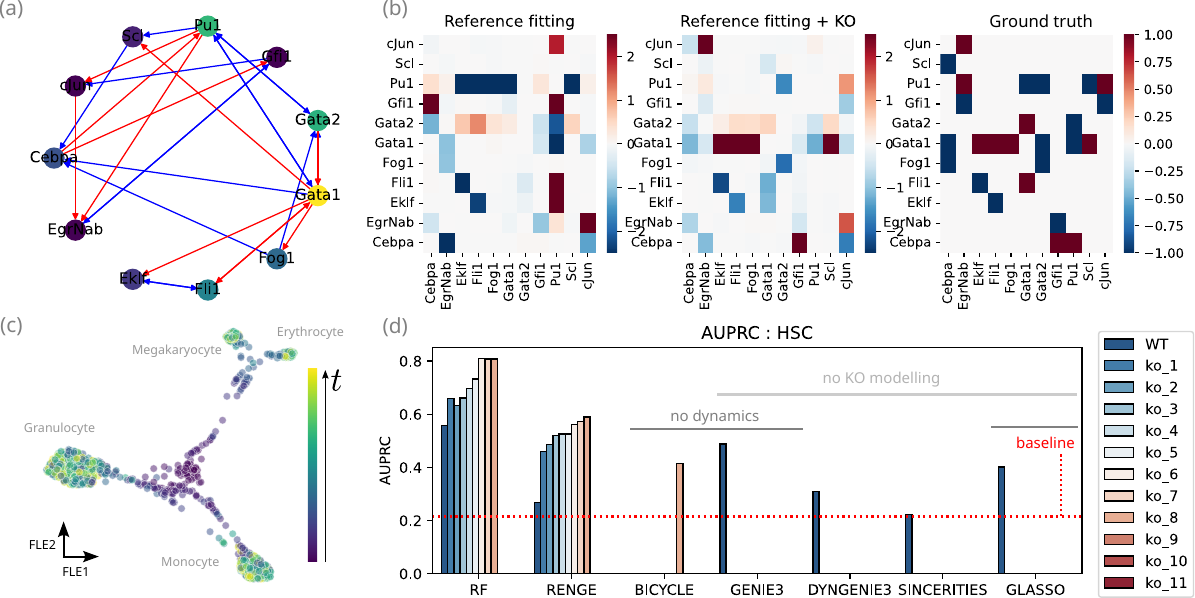}
  \caption{(a) Haematopoietic stem cell (HSC) network from \cite{pratapa2020benchmarking}. Activating (inhibitory) edges are shown in red (blue). (b) Inferred networks from reference fitting without (AUPR = 0.29) and with knockouts (AUPR = 0.73), compared to the ground truth. (c) Force-layout dimensionality reduction of wild-type trajectory coloured by time, showing four stable states. (d) Network inference accuracy (averaged over 10 datasets) as measured by AUPRC using reference fitting with and without knockouts, compared to alternative methods.}
  \label{fig:hsc}
\end{figure}

As a biologically grounded example, we consider a network of 11 TFs involved in haematopoietic stem cell (HSC) differentiation \cite{krumsiek2011hierarchical} (Figure \ref{fig:hsc}(a)), across 5 timepoints. This network produces a dynamics in which cells evolve from a stem-like state characterised by Cebpa$^+$/Gata2$^+$/Pu1$^+$ towards four clusters corresponding to granulocyte, monocyte, erythrocyte and megakaryocyte \cite[Supplementary Figure 8]{pratapa2020benchmarking}. We also simulate knockout trajectories for the top 5 TFs by centrality: $\{ \text{Gata1}, \text{Fli1}, \text{Fog1}, \text{Eklf}, \text{Scl} \}$. As with the trifurcating network, providing reference fitting with knockout information yields a considerable performance improvement (AUPR = 0.73, compared to AUPR = 0.29 without knockouts). Without knockouts, we find that reference fitting performs comparably to or better than competing approaches (Figure \ref{fig:hsc}(c, d)). 
Finally, we find that SINCERITIES struggles again in this example, where the number of timepoints ($T = 5$) is the minimum possible for the method to handle.

\paragraph{CRISPR perturbation time-series}

Ishikawa et al. \cite{ishikawa2023renge} generated a time-series dataset of human induced pluripotent stem cell (iPSC) with CRISPR knockout perturbations of 23 transcription factors and across 4 timepoints. Thus for each knockout, the temporal propagation of the loss of expression is captured. We used the energy distance \cite{rizzo2016energy, peidli2024scperturb}, \rev{a nonparametric distance between general distributions,} to quantify the change in population-level gene expression between each knockout population and the wild type (see Appendix Figure \ref{fig:crispr_edists}). It is clear that there are several TF knockouts which result in large changes in cell state, while other TFs have negligible effects. We therefore selected the top 8 knockouts (ranked by energy distance), $\{ \text{Prdm14}, \text{Pou5f1}, \text{Runx1t1}, \text{Sox2}, \text{Zic2}, \text{Nanog}, \text{Myc}, \text{Zic3} \}$ for further investigation. We show in Figure \ref{fig:crispr}(a) the different perturbed population profiles. The knockouts associated with the largest change in cell state (both visible from the UMAP and in terms of the energy distance) include Pou5f1 and Sox2, both known to be central and canonical regulators implicated in the maintenance of pluripotency and stem cell differentiation \cite{pan2002stem, rizzino2009sox2, chew2005reciprocal}. When Oct4 is knocked out (Figure \ref{fig:crispr}(b)), the time-series captures a progressive shift in cell states as the knockout effect propagates. On the other hand, the wild type population remains in a steady state.

\begin{figure}[h]
  \centering
  \includegraphics[width = \linewidth]{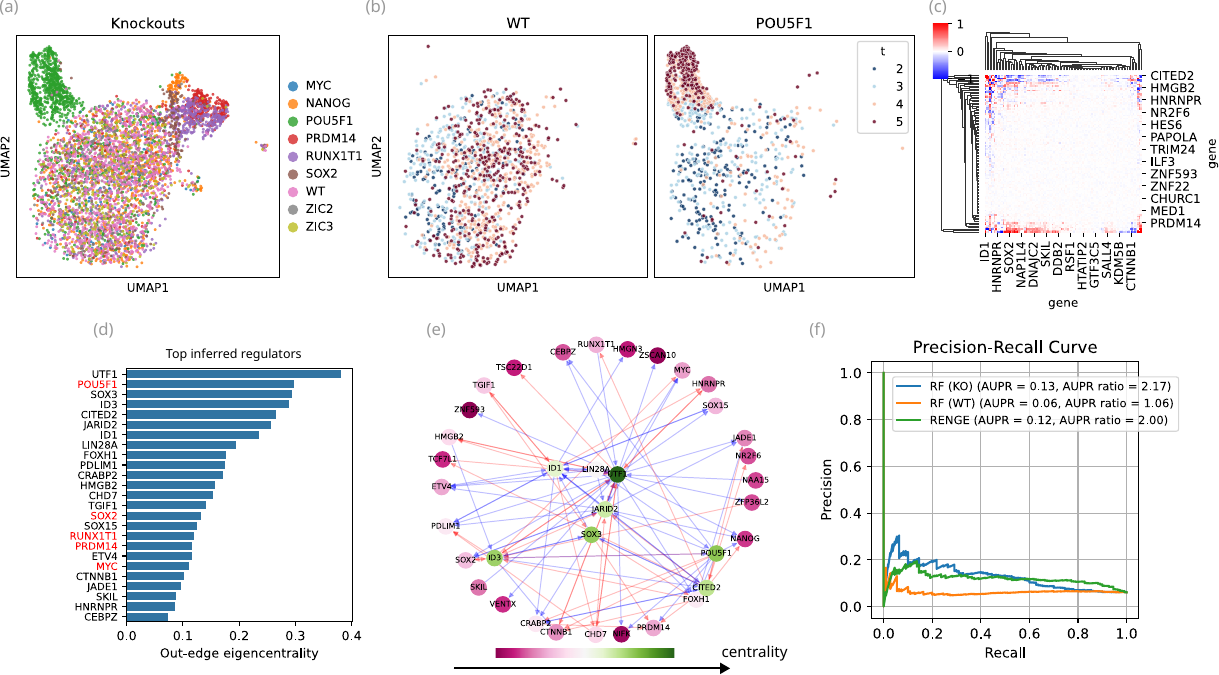}
  \caption{\rev{(a) Single cells from wild type and knockout conditions across timepoints, coloured by knockout condition. (b) Wild-type and Pou5f1 (Oct4) knockout cells coloured by timepoint, showing temporal dynamics of knockout propagation. (c) Inferred network on subset of 103 TFs, shown as a signed adjacency matrix. (d) Top 25 TFs in network inferred using reference fitting with knockouts, by out-edge eigenvector centrality. TFs for which knockout data were used shown in red. (e) Inferred network (thresholded top 2.5\% of edges) coloured by out-edge eigenvector centrality. (f) Precision-recall curves for subnetwork, using ChIP-seq database as reference. }}
  \label{fig:crispr}
\end{figure}

Applying reference fitting to the subset of 103 transcription factors considered in \cite{ishikawa2023renge}, we obtain a $103 \times 103$ directed, signed network of inferred TF-TF interactions (Figure \ref{fig:crispr}(c)). Since there is no definitive ground truth for real biological interaction networks, in Figure \ref{fig:crispr}(d) we show the top 25 TFs ranked by out-edge eigenvector centrality -- genes for which knockout data were provided are shown in red. In Figure \ref{fig:crispr}(e) we also show the network structure filtered for the top 2.5\% of interactions. Pou5f1 (Oct4) is the second highest ranked for centrality, reflecting its known role as a master regulator. Notably, many of the other top-ranked TFs did not have knockout information. This agrees with our simulation findings that reference fitting is able to integrate perturbational and dynamical information for network inference. Among other top ranked regulators, we find Lin28a, which together with Oct4, Sox2 and Nanog were found to be sufficient for reprogramming human somatic cells in a landmark study \cite{yu2007induced}. We emphasise that no knockout information for Lin28a itself was used in this analysis.  In contrast, in the same analysis for the network inferred by RENGE (Appendix Figure \ref{fig:crispr_edists}(b)) there is a clear bias for knockout TFs to have higher centrality
Finally, in Figure \ref{fig:crispr}(f) we calculated precision-recall curves for a subset of 18 transcription factors for which ChIP-seq binding information were available \cite{zou2024chip}. We found that reference fitting performed close to random when only run on wild-type data (AUPR ratio 1.06). Providing only a relatively small number of knockouts (8 out of 103 TFs considered) is sufficient to double the prediction performance (AUPR ratio 2.17). We remark that in this hiPSC dataset, the wild-type cell population is stationary (see \cite[Supplementary Note 2]{ishikawa2023renge} and \ref{fig:crispr}(c)) so the poor result in the wild-type case is to be expected.

\section{Discussion}

Motivated by information-rich time-resolved and perturbation single cell experiments, we propose a computational approach for joint trajectory and network inference. Our approach draws inspiration from the theory of entropy regularised optimal transport and inference for linear dynamical systems. We posit that a least action principle should be satisfied: the most likely system should be the one that minimises the total action of the observed dynamics. Using simulated data from both linear (Ornstein-Uhlenbeck) and non-linear (synthetic and biological) stochastic systems, we demonstrate that our approach is able to leverage both transient dynamics as well as perturbation information to infer better trajectories and networks. In particular, we find that perturbing a fraction of genes greatly improves network inference compared to only using unperturbed dynamics. We demonstrate the applicability of our approach to real biological time-series data with perturbations and show that the inferred networks agree with prior knowledge. 

In future work, we will address settings where a combination of steady-state and time-series data are available -- for instance where the wild-type system is observed across time, but perturbations are observed at only a single snapshot. Other potential extensions of our approach include modelling non-autonomous systems by allowing the networks to vary over time \cite{wang2024wendy}, as well as to utilise additional dynamical information such as RNA velocity or metabolic labelling \cite{zhang2023learning, xu2023dissecting}. 
Finally, theoretical results are a important direction to be investigated: while some theoretical work has been done for inference in Ornstein-Uhlenbeck processes at steady state \cite{dettling2023identifiability}, the case of transient dynamics without longitudinal measurements is less studied. 

\section*{Acknowledgements}

SZ gratefully acknowledges insightful discussions with, and support from, Dr. Xiaojie Qiu (Stanford) and funding from the Australian Government Research Training Program, Elizabeth and Vernon Puzey Scholarship, Prof. Maurice H Belz Fund and the School of Mathematics and Statistics at the University of Melbourne. 

\newpage 
\bibliographystyle{plain}
\bibliography{references}

\newpage 
\appendix 

\section{Supplementary Material}

\subsection{Supplementary figures}

\begin{figure}[h]
  \centering
  \begin{subfigure}{0.49\linewidth}
    \includegraphics[width = \linewidth]{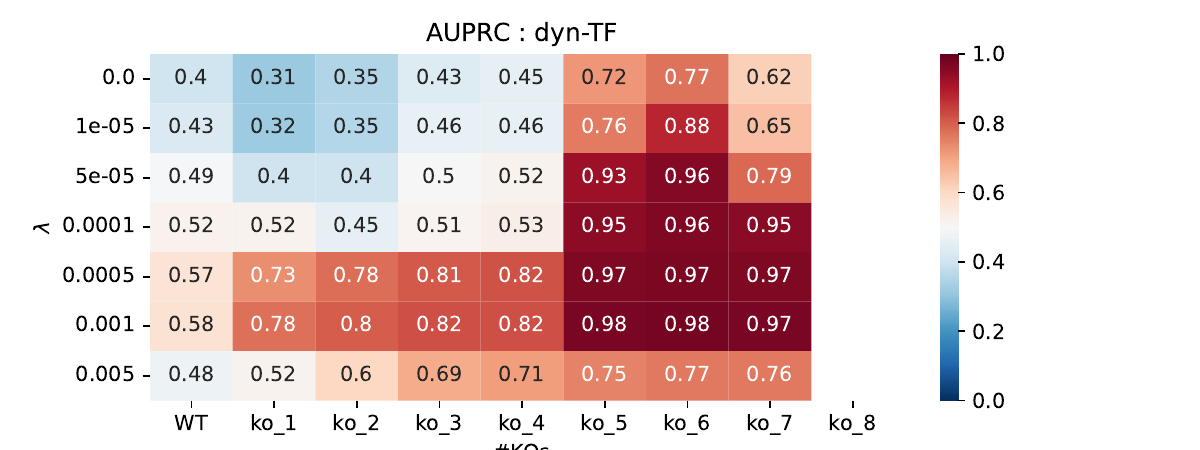}
    \caption{Trifurcating}
  \end{subfigure}
  \begin{subfigure}{0.49\linewidth}
    \includegraphics[width = \linewidth]{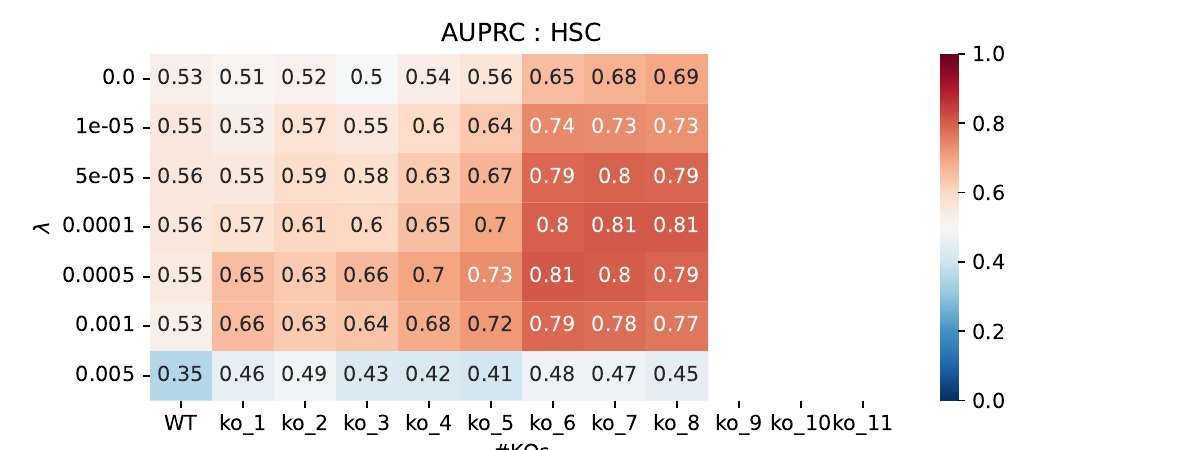}
    \caption{HSC}
  \end{subfigure}
  \caption{AUPRC scores for reference fitting in (a) trifurcating and (b) HSC systems, with different numbers of knockouts and different regularisation strengths $\lambda$}
  \label{fig:auprc_all}
\end{figure}

\begin{figure}[h]
  \centering 
  \begin{subfigure}[b]{0.245\linewidth}
    \includegraphics[width = \linewidth]{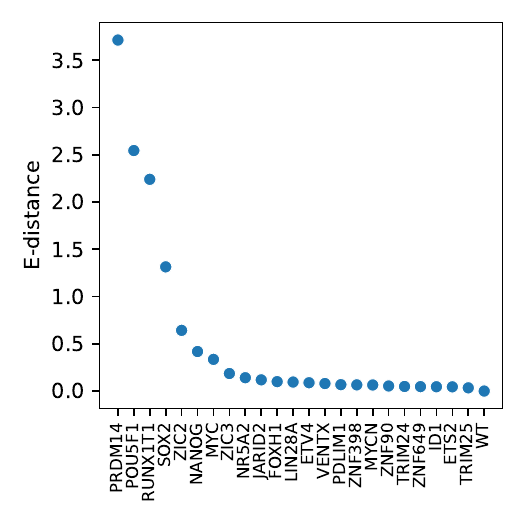}
    \caption{}
  \end{subfigure}
  \begin{subfigure}[b]{0.245\linewidth}
    \includegraphics[width = \linewidth]{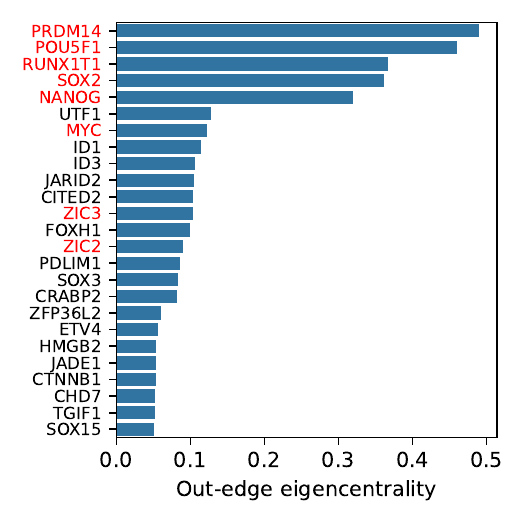}
    \caption{}
  \end{subfigure}
  \caption{(a) Knockouts ranked by energy distances between knockout population and wild type. (b) Out-edge eigenvector centrality for network inferred by RENGE. }
  \label{fig:crispr_edists}
\end{figure}

\subsection{\rev{Solving the reference fitting problem}}

The reference fitting problem between two distributions $(\mu, \mu')$ at times $t = 0, 1$ is
\begin{equation}
    \min_{A, \pi \in \mathbb{R}^{d \times d} \times \Cc(\mu, \mu')} \sigma^2 \KL(\pi | K^{\sigma}_{A}) + \Rc(A),
    \label{eq:rf_objective_appendix}
\end{equation}
with constraint set $(A, \pi) \in \mathbb{R}^{d \times d} \times \Cc(\mu, \mu')$. 
As we pointed out earlier, this objective is a non-convex problem due to the dependence of $K_A^\sigma$ on $A$. Furthermore, optimising in $A$ becomes numerically difficult when using the exact O-U transition kernel due to the connection between the drift and covariance (and thus the need to invert the covariance matrix becomes a difficulty). We remark that the setting where multiple time-points are available can be treated by considering the sum of similar terms as in \eqref{eq:rf_objective_appendix}. 

\begin{prop}[Existence of minimisers]
  Consider the problem \eqref{eq:rf_objective_appendix} where $K_A^\sigma$ is taken to be the approximate reference kernel, i.e.
  \[
    K_A^\sigma(x, x') \propto \exp\left( - \frac{ \| e^{A} x - x' \|_2^2}{2\sigma^2} \right). 
  \]
  If $\Rc$ is bounded below and $\Rc(A) \to \infty$ whenever $\| A \|_F \to \infty$, then \eqref{eq:rf_objective_appendix} has at least one minimiser. 
\end{prop}
\begin{proof}
  For the approximate transition kernel, one has that
  \[
    -\log K_A^\sigma (x, x') = \frac{1}{2\sigma^2} \| e^{A} x - x' \|_2^2 + \mathrm{const}.  \\ 
  \]
  Up to a constant that is independent of $A, \pi$, then, the objective \eqref{eq:rf_objective_appendix} is equal to
  \begin{equation}
    \min_{A, \pi} \sigma^2 \sum_{ij} \pi_{ij} \log\pi_{ij} + \frac{1}{2} \sum_{ij} \pi_{ij} \| e^{A} x_i - x_j \|_2^2 + \Rc(A).
    \label{eq:rf_objective_rearranged}
  \end{equation}
  The first term is an entropy term and is bounded below, and the second term is non-negative. Let $\Rc$ be a coercive regulariser, i.e. suppose that $\Rc(A) \to +\infty$ whenever $\| A \|_F \to \infty$.
  The objective \eqref{eq:rf_objective_rearranged} is then continuous and bounded below on $\mathbb{R}^{d \times d} \times \Cc(\mu, \mu')$, coercive on $\mathbb{R}^{d \times d}$, and $\Cc(\mu, \mu')$ is compact. We conclude existence of a global minimiser $(A^\star, \pi^\star)$. 
\end{proof}

\begin{rmk}[Alternating scheme]
  Let $F(A, \pi)$ be the objective of \eqref{eq:rf_objective_appendix} for which we seek a local minimum. Let $(A_0, \pi_0) \in \mathbb{R}^{d \times d} \times \Cc(\mu, \mu')$ be given and consider the alternating minimisation scheme
  \begin{align*}
    A_{k+1} &\gets \arg\min_{A} F(A, \pi_k) \\
    \pi_{k+1} &\gets \arg\min_{\pi} F(A_{k+1}, \pi). 
  \end{align*}
  This generates a sequence $(A_k, \pi_k)$ for $k \ge 0$ such that $F(A_k, \pi_k)$ is nonincreasing, by construction. Since $F$ is bounded below, the sequence of objectives $F(A_k, \pi_k)$ must converge in its value. However, since $F$ is non-convex in $A$, the sequence $(A_k, \pi_k)$ need not converge. 
\end{rmk}

\begin{rmk}
  If $\Kc$ were a convex family of densities, issues of convexity could be alleviated since then the reference fitting problem has the form
  \[
    \min_{(K, \pi) \in \Kc \times \Cc(\mu, \mu')} \KL( \pi | K ),
  \]
  and $\KL$ is jointly convex (but not strictly) in its arguments. In this case, reference fitting would amount to (a variant of) alternating projections onto convex sets. 
  In the case of Ornstein-Uhlenbeck reference processes, however, reference densities are multivariate Gaussian which do not form a convex set in the space of densities.  
\end{rmk}

\begin{rmk}[Convergence of alternating scheme]
  The update in $\pi$ is (strongly) convex. The update in $A$ is non-convex. Assume that $\Rc(A)$ is convex. We adopt the scheme proposed by \cite{xu2013block} for the problem of block-coordinate minimisation of 
  \begin{align*}
    A_{k+1} &\gets \arg\min_{A} F(A, \pi_k) + \beta \| A - A_k \|_F^2 \\
    \pi_{k+1} &\gets \arg\min_{\pi} F(A_{k+1}, \pi). 
  \end{align*}
  Note the additional proximal term added to the non-convex block for $A$, which is required for convergence to a critical point.

  As previously we note that the objective $F(A, \pi)$ can be re-written in the form:
  \[
    \min_{(A, \pi) \in \mathbb{R}^{d \times d} \times \Cc(\mu, \mu')} \left[ \sigma^2 \sum_{ij} \pi_{ij} \log\pi_{ij} + \frac{1}{2} \sum_{ij} \pi_{ij} \| e^{A} x_i - x_j \|_2^2 \right] + \Rc(A).
\]
The feasible set for $(A, \pi)$ is closed and convex, the first two terms are smooth in $(A, \pi)$, strongly convex in $\pi$ and non-convex in $A$. The last term $\Rc(A)$ is convex in $A$ and possibly non-smooth.
This falls into the framework of \cite{xu2013block} which proves global convergence to a critical point under some additional technical conditions. The the update in $\pi$ is handled by a standard block minimisation, but the update in $A$ must be handled by a proximal update (see Eq. 1.3b in \cite{xu2013block}), i.e.
\[
  A_{k+1} \gets \arg\min_{A} F(A, \pi_k) + \beta \| A - A_k \|_F^2 . 
\]
\end{rmk}

\subsection{Datasets and preprocessing}

Code to reproduce results can be found at \url{https://github.com/zsteve/referencefitting}.

\paragraph{8-D non-equilibrium OU process}

Particles were simulated following an Ornstein-Uhlenbeck process \eqref{eq:linear_sde} with
\[
  A = 1.25 \begin{bmatrix}
    0 & 1  & 0  & 0  & 0  & 0  & 0  & 0 \\
    0  & 0 & 1  & 0  & 0  & 0  & 0  & 0 \\
    0  & 0  & 0 & 1  & 0  & 0  & 0  & 0 \\
    0  & 0  & 0  & 0 & 1  & 0  & 0  & 0 \\
    0  & 0  & 0  & 0  & 0 & 1  & 0  & 0 \\
    0  & 0  & 0  & 0  & 0  & 0 & 1  & 0 \\
    0  & 0  & 0  & 0  & 0  & 0  & 0 & 1 \\
    -1  & 0  & 0  & 0  & 0  & 0  & 0  & 0 
  \end{bmatrix} - 1.25 I 
\]
and $\sigma = 0.05I$. We independently sampled 10 snapshots each with 100 particles, evenly spaced between $t = 0$ and $t = 10$ with $x_0 = 0.25e_1 + \Nc(0, 0.05^2I)$.

\paragraph{Synthetic trifurcating and HSC trajectories}

We used BoolODE \cite{pratapa2020benchmarking} to simulate 1000 cells independently from each trajectory. To generate time-resolved snapshots, the simulation time was binned into $T = 5$ discrete timepoints. Simulated expression values were log-transformed before being used as input for downstream tasks. In order to generate each knockout trajectory, the boolean rules were modified such that the knocked-out gene is only subject to self-activation, and setting its initial expression level to zero.

\rev{To systematically examine the performance for varying numbers of knockouts, 10 independently generated simulated datasets (each with 1000 cells) for each trajectory are used. Starting from WT-only, we progressively add knockout trajectories for genes in order of their out-edge eigencentrality, from highest to lowest. In order:
\begin{itemize}
\item Trifurcating: g3, g6, g4, g5, g2, g7, g8
\item HSC: Gata1, Fli1, Fog1, Eklf, Scl, Gfi1, EgrNab, cJun.
\end{itemize}
For each instance, we run reference fitting with different choices of the LASSO regularisation hyperparameter: $\lambda \in \{0, 10^{-5}, 5\times 10^{-5}, 10^{-4}, 5\times 10^{-4}, 10^{-3}, 5\times 10^{-3} \}$. We measure the accuracy of the inferred network in terms of AUPRC and for each instance we report the best score across the different regularisation hyperparameters. The full set of results, averaged across the 10 datasets, are shown in Figure \ref{fig:auprc_all}. 
For comparison, we also run the following methods:
\begin{itemize}
  \item RENGE \cite{ishikawa2023renge} with the same time-series and knockout combinations as input. RENGE is able to model both temporal and knockout data. 
  \item BICYCLE \cite{rohbeck2024bicycle} with all cells and all knockouts as input. BICYCLE models knockout data but without temporal resolution. 
  \item GENIE3 \cite{huynh2010inferring} with wild-type cells. GENIE3 models single cell data without knockouts or temporal resolution.
  \item dynGENIE3 \cite{huynh2018dyngenie3} with wild-type trajectory. dynGENIE3 models temporally resolved single cell data.
  \item SINCERITIES \cite{papili2018sincerities} with wild-type trajectory. SINCERITIES models temporally resolved single cell data.
  \item Graphical LASSO (GLASSO) \cite{friedman2008sparse} with wild-type cells. 
  \end{itemize}
}

\paragraph{Single-cell CRISPR perturbation time-series}

The raw count data for \cite{ishikawa2023renge} were retrieved from the Gene Expression Omnibus database (accession GSE213069). Columns corresponding to gRNAs were removed, then counts were normalised using the \texttt{scanpy.pp.normalize\_total} function with default options, log-transformed. Prior to dimensionality reduction, highly variable genes were selected using \texttt{scanpy.pp.highly\_variable\_genes}. As an input to network inference, we considered the set of 103 TFs from \cite{ishikawa2023renge} and considered only cells that received a single knockout. To construct the ChIP-seq reference, we obtained experimental binding information from ChIP-atlas \cite{zou2024chip} for the following TFs:
Chd7, Ctnnb1, Dnmt1, Foxh1, Jarid2, Kdm5b, Med1, Myc, Nanog, Nr5a2, Pou5f1, Prdm14, Sall4, Sox2, Tcf3, Tcf7l1, Ubtf, Znf398 with a 1kb window.

\end{document}